\documentclass[journal]{IEEEtran}
\usepackage{graphicx,amsmath,amssymb}
\usepackage{cite}
\usepackage{subfigure}
\usepackage{citesort}
\usepackage{fancyhdr}
\usepackage{mdwmath}
\usepackage{mdwtab}
\usepackage{balance}
\usepackage{xcolor}
\usepackage{bm}
\usepackage{amsthm}
\usepackage{algorithm}
\usepackage{algorithmic}
\usepackage{multirow}
\usepackage{flafter}
\usepackage{epstopdf}
\newtheorem{remark}{Remark}

\newtheorem{theorem}{Theorem}

\newtheorem{lemma}{Lemma}

\newtheorem{corollary}{Corollary}

\newtheorem{assumption}{Assumption}

\begin{document}
\title{Secrecy Analysis of Ambient Backscatter NOMA Systems under I/Q Imbalance}

\author{Xingwang~Li,~\IEEEmembership{Senior Member,~IEEE,}
        Mengle~Zhao,~\IEEEmembership{Student Member,~IEEE,}
        Yuanwei~Liu,~\IEEEmembership{Senior Member,~IEEE,}
        Lihua~Li,~\IEEEmembership{Member,~IEEE,}
        Zhiguo~Ding,~\IEEEmembership{Fellow,~IEEE,}
        and Arumugam~Nallanathan,~\IEEEmembership{Fellow,~IEEE}
\thanks{X. Li and M. Zhao are with the School of Physics and Electronic Information Engineering, Henan Polytechnic University, Jiaozuo, China (email:lixingwangbupt@gmail.com, zhaomenglehpu@163.com).}
\thanks{Y. Liu and A. Nallanathanis are with the School of Electronic Engineering and Computer Science, Queen Mary University of London, London, UK, London, UK (email:\{yuanwei.liu, a.nallanathan\}@qmul.ac.uk).}
\thanks{L. Li is with the State Key Laboratory of Networking and Switching Technology, Beijing University of Posts and Telecommunications, China (email:lilihua@bupt.edu.cn).}
\thanks{Z. Ding is with the School of Electrical and Electronic Engineering, The University of Manchester, Manchester, UK (email: zhiguo.ding@manchester.ac.uk).}
}

\maketitle
\nocite{Li2018}
\begin{abstract} We investigate the reliability and security of the ambient backscatter (AmBC) non-orthogonal multiple access (NOMA) systems, where the source aims to communication with two NOMA users in the presence of an eavesdropper. We consider a more practical case that nodes and backscatter device (BD) suffer from in-phase and quadrature-phase imbalance (IQI). More specifically, exact analytical expressions for the outage probability (OP) and the intercept probability (IP) are derived in closed-form. Moreover, the asymptotic behaviors and corresponding diversity orders for the OP are discussed. Numerical results show that: 1) Although IQI reduces the reliability, it can enhance the security. 2) Compared with the traditional orthogonal multiple access (OMA) system, the AmBC-NOMA system can obtain better reliability when the signal-to-noise (SNR) ratio is low; 3) There are error floors for the OP because of the reflection coefficient $\beta$.

\end{abstract}


\section{Introduction}
Non-orthogonal multiple access (NOMA) has been identified as one of the key technologies of the fifth-generation (5G) mobile networks due since it has the advantages of high spectrum efficiency, massive connection and low latency \cite{1}. Different from orthogonal multiple access (OMA), the dominant feature of NOMA is to ensure multiple users to occupy the the same frequency/time resources by power domain multiplexing. These advantages are achieved by employing superposed code (SC) at the transmitter and successive interference cancellation (SIC) at the receiver \cite{2}. Moreover, NOMA can ensure the fairness by allocating more power to the weak users.

On a parallel avenue, ambient backscatter communication (AmBC) is known as a potential technology with high spectrum- and energy-efficient for the green Internet-of-Things (IoT) and has been attracted widespread attention in academia and industry \cite{3}. In general, the AmBC system consists of three parts: ambient radio-frequency (RF) source, backscatter device (BD) and reader. In this respect, the cooperative AmBC system was designed in \cite{4} to allow the reader to recover information from both BD and RF source, and the analytical closed-form expressions for the bit error rate (BER) of the maximum-likelihood (ML), suboptimal linear detector and SIC detectors were derived. In \cite{5}, Guo \emph{et al}. proposed a NOMA-assisted AmBC system to support massive BD connections.

Due to the broadcast characteristic of wireless environments, it is difficult to ensure secure communication for the wireless networks without being eavesdropped by the un-authored receivers. Physical layer security (PLS) has been known as an effective way to improve the security of communication systems, which has sparked a great deal of research interests \cite{6,7,8,9}. The authors of \cite{6} investigated the reliability and security of multi-relay networks in terms of the outage probability (OP) and the intercept probability (IP). Regarding NOMA systems, the authors proposed an enhanced security scheme to against full-duplex proactive eavesdropping \cite{7}. To enhance the security of AmBC systems, the authors in \cite{8} designed an optimal tag selection scheme of the multi-tag AmBC systems. However, the common characteristic of above works are assumed that the transceiver RF front-ends were equipped with ideal components. Unfortunately, the practical RF components are prone to in-phase and quadrature-phase imbalance (IQI) due to mismatch and manufacturing non-idealities \cite{9}, which limit the overall system performance. Therefore, it is of great practical significance to study the secrecy performance of the AmBC NOMA systems with IQI.

Although the PLS of the NOMA systems has been extensively studied, to the best of our knowledge, the impact of IQI on the PLS performance of the AmBC NOMA system has not yet been investigated. This work aims to bridge this gap and investigates the effects of IQI on the PLS of the AmBC NOMA systems, where the eavesdropper can intercept the signals from both the source and reflected signals from BD. Specifically, we study the reliability and security by deriving the analytical closed-form expressions of the OP and the IP for the far user, the near user and the BD, respectively. Furthermore, in order to obtain more insights, the asymptotic behavior for the OP at high signal-to-noise (SNR) is explored, as well as the diversity order. The results have shown that although IOI has deleterious effects on reliability, it can enhance the secure performance of the considered system. Moreover, the proposed AmBC system can provide highly secure communication for the BD.

\begin{figure}[!t]
\setlength{\abovecaptionskip}{0pt}
\centering
\includegraphics [width=2.5in]{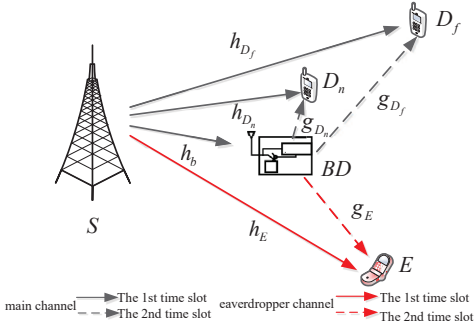}
\caption{System model}
\label{fig2}
\end{figure}

\section{System Model}\label{sec2}
We consider a downlink AmBC NOMA system as illustrated in Fig. 1, which consists of one source $S$, one BD, one far user ($D_f$), one near user ($D_n$) and one eavesdropper ($E$). We assume that: i) All the nodes are equipped with a single antenna; ii) The transceiver RF front-ends of the all nodes suffer from IQI. For convenience, we have
${h_{{i}}} \sim {\rm{{\cal C}{\cal N}}}\left( {0,{\lambda _k}} \right)$, ${g_{{i}}} \sim {\rm{{\cal C}{\cal N}}}\left( {0,{\lambda _k}} \right)$, $i \in \{D_n, D_f, E, \}$, $k \in \{1, 2, 3, 4, 5, 6\}$ and ${h_{{b}}} \sim {\rm{{\cal C}{\cal N}}}\left( {0,{\lambda _7}} \right)$,

As in \cite{9}, IQI is modeled as phase and/or amplitude imbalance between the in-phase (I) and quadrature-phase (Q) branches. Thus, the time-domain baseband TX signals with IQI at $S$ can be expressed as
\begin{equation}
\label{2}
{y_s} = {\mu _{{t_S}}}{x_s} + {\nu _{{t_S}}}{\left( {{x_s}} \right)^*},
\end{equation}
where ${x_s} = {\sqrt {{a_1}{P_s}} {x_1}} + {\sqrt {{a_2}{P_s}} {x_2}}$, ${P_S}$ is the transmit power at $S$; ${a_1}$ and ${a_2}$ are the power allocation coefficients with ${a_1} + {a_2} \!= \!1$ and ${a_1} \!<\! {a_2}$, respectively; ${x_1}$ and ${x_2}$ are corresponding transmitted signals of $D_f$ and $D_n$ with $E\left( {{{\left| {{x_1}} \right|}^2}} \right)\! =\! E\left( {{{\left| {{x_2}} \right|}^2}} \right)\! =\! 1$, where $E\left(  \cdot \right)$ and ${\left(  \cdot  \right)^*}$ are the expectation and conjugation operations, respectively. Moreover, ${\mu _{t_s}} = \frac{1}{2}\left( {1 + {\varsigma _t}\exp \left( {j{\phi _t}} \right)} \right)$, ${\nu _{t_s}} = \frac{1}{2}\left( {1 - {\varsigma _t}\exp \left( { - j{\phi _t}} \right)} \right)$, where ${\varsigma _{t_s}}$ and ${\phi _{t_s}}$ denote the TX amplitude and phase mismatch levels, respectively.

The BD backscatters the received signal to $D_n$ with its own signal $c\left( t \right)$ with $E\left( {{{\left| {c\left( t \right)} \right|}^2}} \right) = 1$. Two types of received signals at $i$ are transmitted from $S$ and backscattered from BD. Considering IQI at transceivers, the received signals at $i$ are given by
\begin{align}\nonumber
\label{3}
{y_i} &= {\mu _{{r_i}}}\left\{ {\beta {y_{BD}}{g_i}\left[ {{\mu _{{t_{BD}}}}c\left( t \right) + {\nu _{{t_{BD}}}}{{\left( {c\left( t \right)} \right)}^*}} \right] + {h_i}{y_s} + {n_i}} \right\} \\
+& {\nu _{{r_i}}}{\left\{ {\beta {y_{BD}}{g_i}\left[ {{\mu _{{t_{BD}}}}c\left( t \right) \!+ \! {\nu _{{t_{BD}}}}{{\left( {c\left( t \right)} \right)}^*}} \right] \!+ \!{h_i}{y_s} \!+ \!{n_i}} \right\}^*},
\end{align}
where $\beta $ is a complex reflection coefficient used to normalize $c\left( t \right)$; ${n_i} \sim {\rm{{\cal C}{\cal N}}}\left( {0,{N_0}} \right)$ is the complex additive white Gaussian noise (AWGN); ${h_i}$ and ${g_i}$ are the channel coefficients $S \to i$ and $BD \to i$, respectively. Moreover, ${\mu _{r_i}} = \frac{1}{2}\left( {1 + {\varsigma _{r_i}}\exp \left( { - j{\phi _{r_i}}} \right)} \right)$, ${\nu _{r_i}} = \frac{1}{2}\left( {1 - {\varsigma _{r_i}}\exp \left( {j{\phi _{r_i}}} \right)} \right)$, where ${\varsigma _{r_i}}$ and ${\phi _{r_i}}$ are the RX amplitude and phase mismatch levels, respectively. ${y_{BD}} = {\mu _{{r_{BD}}}}{h_b}{y_s} + {\nu _{{r_{BD}}}}{\left( {{h_b}{y_s}} \right)^*}$ is the received signals at $BD$, where ${h_b}$ is the channel coefficient of $S \to BD$. The received signal-to-interference-plus-noise ratio (SINR) for $D_f$ to decode the signal ${x_2}$ can be obtained as
\begin{equation}
\label{4}
\gamma _{{D_f}}^{{x_2}} = \frac{{{\xi _{{D_f}}}{a_2}{\rho _{{D_f}}}\gamma }}{{{\rho _{{g_{{D_f}}}}}{\rho _b}\left( {{A_{{D_f}}} + {Q_{{D_f}}}} \right)\gamma  + {\rho _{{D_f}}}{C_{{D_f}}}\gamma  + {D_{{D_f}}}}},
\end{equation}

According to the NOMA principle, $D_n$ and $E$ can decode signal ${x_2}$, ${x_1}$, and $c\left( t \right)$ in turn with the aid of SIC. Then, the received SINR of $i$ can be given as
\begin{equation}
\label{5}
\gamma _{i}^{{x_2}} = \frac{{{\xi _i}{a_2}{\rho _{{D_n}}}\gamma }}{{{\rho _{{g_{i}}}}{\rho _b}\left( {{A_{i}} + {Q_{i}}} \right)\gamma  + {\rho _{i}}{C_{i}}\gamma  + {D_{i}}}},
\end{equation}
\begin{equation}
\label{6}
\gamma _{i}^{{x_1}} \!=\! \frac{{{\xi _i}{a_1}{\rho _{i}}\gamma }}{{{\rho _{{g_{i}}}}{\rho _b}\!\left( \!{{A_{i}} \!+\! {Q_{i}}} \!\right)\gamma +{\rho _{i}}\!\left( \!{{a_2}{B_{i}} \!+\! {M_{i}}} \!\right)\gamma  \!+\! {D_{i}}}},
\end{equation}
\begin{equation}
\label{7}
\gamma _{i}^{c\left( t \right)} = \frac{{{Q_{i}}{\rho _{{g_{i}}}}{\rho _b}\gamma }}{{{\rho _{{g_{i}}}}{\rho _b}{A_{i}}\gamma  + {\rho _{i}}\left( {{B_{i}} + {M_{i}}} \right)\gamma  + {D_{i}}}},
\end{equation}
where $\gamma  \!= \!{{{P_s}} \mathord{\left/
 {\vphantom {{{P_s}} {{N_0}}}} \right.
 \kern-\nulldelimiterspace} {{N_0}}}$ represents the transmit SNR at $S$; ${\rho _i} \!= \!{\left| {{h_i}} \right|^2}$, ${\rho _{{g_i}}} \!= \!{\left| {{g_i}} \right|^2}$, ${\rho _b} \!=\! {\left| {{h_b}} \right|^2}$, ${A_i} \!= \!{\left| {{\mu _{{r_i}}}} \right|^2}{\beta ^2}{\left| {{\mu _{{r_{BD}}}}} \right|^2}{\left| {{\nu _{{t_{BD}}}}} \right|^2}\left( \!{{{\left| \!{{\mu _{{t_S}}}} \!\right|}^2} \!+\! {{\left| \!{{\nu _{{t_S}}}} \!\right|}^2}} \!\right) \!+\! {\left| {{\mu _{{r_i}}}} \right|^2}{\beta ^2}{\left| {{\nu _{{r_{BD}}}}} \right|^2}\\
 \times{\left| {{\nu _{{t_{BD}}}}} \right|^2}\left( {{{\left| {{\mu _{{t_S}}}^*} \right|}^2} \!+\! {{\left| {{\nu _{{t_S}}}^*} \right|}^2}} \right) \!+\! {\left| {{\nu _{{r_i}}}} \right|^2}{\beta ^2}{\left| {{\mu ^*}_{{r_{BD}}}} \right|^2}{\left| {{\mu ^*}_{{t_{BD}}}} \right|^2}\times\\
 \!\left( \!{{{\left| \!{{\mu _{{t_S}}}^*} \right|}^2} \!+\! {{\left| {{\nu _{{t_S}}}^*} \right|}^2}} \right) \!+ \!{\left| {{\nu _{{r_i}}}} \right|^2}{\beta ^2}{\left| {{\nu ^*}_{{r_{BD}}}} \right|^2}{\left| {{\mu ^*}_{{t_{BD}}}} \right|^2}\!\left(\! {{{\left| {{\mu _{{t_S}}}} \right|}^2} \!+\! {{\left| {{\nu _{{t_S}}}} \right|}^2}} \!\right)$, ${Q_i} \!= \!{\left| {{\mu _{{r_i}}}} \right|^2}{\beta ^2}{\left| {{\mu _{{r_{BD}}}}} \right|^2}{\left| {{\mu _{{t_{BD}}}}} \right|^2}\left( \!{{{\left| \!{{\mu _{{t_S}}}} \!\right|}^2} \!+\! {{\left| \!{{\nu _{{t_S}}}} \!\right|}^2}} \!\right) \!+\! {\left| {{\mu _{{r_i}}}} \right|^2}{\beta ^2}{\left| {{\nu _{{r_{BD}}}}} \right|^2}\\
 \times{\left| {{\mu _{{t_{BD}}}}} \right|^2}\left( {{{\left| {{\mu _{{t_S}}}^*} \right|}^2} \!+\! {{\left| {{\nu _{{t_S}}}^*} \right|}^2}} \right) \!+\! {\left| {{\nu _{{r_i}}}} \right|^2}{\beta ^2}{\left| {{\mu ^*}_{{r_{BD}}}} \right|^2}{\left| {{\nu ^*}_{{t_{BD}}}} \right|^2}\times\\
 \!\left( \!{{{\left| \!{{\mu _{{t_S}}}^*} \right|}^2} \!+\! {{\left| {{\nu _{{t_S}}}^*} \right|}^2}} \right) \!+ \!{\left| {{\nu _{{r_i}}}} \right|^2}{\beta ^2}{\left| {{\nu ^*}_{{r_{BD}}}} \right|^2}{\left| {{\nu ^*}_{{t_{BD}}}} \right|^2}\!\left(\! {{{\left| {{\mu _{{t_S}}}} \right|}^2} \!+\! {{\left| {{\nu _{{t_S}}}} \right|}^2}} \!\right)$, ${\xi _i} = {\left| {{\mu _{{r_i}}}} \right|^2}{\left| {{\mu _{{t_S}}}} \right|^2} + {\left| {{\nu _{{r_i}}}} \right|^2}{\left| {{\nu _{{t_S}}}^*} \right|^2}$, ${B_i} = {\left| {{\mu _{{r_i}}}{\mu _{{t_S}}} - 1} \right|^2} + {\left| {{\nu _{{r_i}}}} \right|^2}{\left| {{\nu _{{t_S}}}^*} \right|^2}$, ${C_i} = {a_1}{\xi _i} + {M_i}$, ${D_i} = {\left| {{\mu _{{r_i}}}} \right|^2} + {\left| {{\nu _{{r_i}}}} \right|^2}$, ${M_i} = {\left| {{\mu _{{r_i}}}} \right|^2}{\left| {{\nu _{{t_S}}}} \right|^2} + {\left| {{\nu _{{r_i}}}} \right|^2}{\left| {{\mu _{{t_S}}}^*} \right|^2}$.

\section{Performance Analyses}\label{sec3}
This section evaluates the reliability and security of the AmBC NOMA systems by deriving the analytical expressions for the OP and the IP.\footnote{The reliability and security are another metrics to characterize the PLS of wireless communication systems without using any secrecy coding, which are formulated by the OP and the IP \cite{7579030}.}
Moreover, we examine asymptotic outage behavior and the diversity order in the high SNR region.

\subsection{Outage Probability Analysis}

\emph{1) Outage Probability for $D_f$}

According to NOMA principle, the outage event occurs at $D_f$ when $D_f$ cannot successfully decode ${x_2}$. Thus, the OP at $D_f$ can be represented as
\begin{equation}
\label{11}
P_{out}^{{D_f}} = 1 - {{\rm{P}}_r}\left( {\gamma _{{D_f}}^{{x_2}} > {\gamma _{th2}}} \right),
\end{equation}
where ${\gamma _{th2}}$ is the target rate at $D_f$.

\begin{theorem}
For Rayleigh fading channels, the analytical expression for the OP of the far user can be obtained as\footnote{The ideal and non-ideal results of OP and IP for the far user and near user can be written in an unified expression in (8), (10), (12), (13), (23) and (24).}
\begin{equation}
\label{12}
P_{out}^{{D_f}} = 1 + {\Delta_1 }{e^{{\Delta_1 } - \frac{{{D_{{D_f}}}{A_1}}}{{{\lambda _2}}}}}{\rm{Ei}}\left( { - {\Delta_1 }} \right),
\end{equation}
\end{theorem}
\noindent where ${A_1} = {{{\gamma _{th2}}} \mathord{\left/ {\vphantom {{{\gamma _{th2}}} {\left( {{\xi _{{D_f}}}{a_2} - {C_{{D_f}}}{\gamma _{th2}}} \right)}}} \right.
 \kern-\nulldelimiterspace} {\left( {{\xi _{{D_f}}}{a_2} - {C_{{D_f}}}{\gamma _{th2}}} \right)}}$, ${\Delta_1 } \!=\! {{{\lambda _2}} \mathord{\left/ {\vphantom {{{\lambda _2}} {\left( {{\lambda _7}{\lambda _5}{A_1}\left( {{A_{{D_f}}} + {Q_{{D_f}}}} \right)} \right)}}} \right. \kern-\nulldelimiterspace} {\left( {{\lambda _7}{\lambda _5}{A_1}\left( {{A_{{D_f}}} + {Q_{{D_f}}}} \right)} \right)}}$, and ${\rm{Ei}}\left( x \right) = \int_{ - \infty }^x {\frac{{{e^\rho }}}{\rho }} d\rho $ is the exponential integral function \cite{12}.

\begin{proof}
By substituting (\ref{4}) into (10), the OP of $D_f$ can be expressed as
\begin{align}\nonumber
\label{12}
&P_{out}^{{D_f}} \!=\! \int_{{A_1}\left( \!{{\rho _{{g_{{D_f}}}}}{\rho _b}\left( \!{{A_{{D_f}}} \!+\! {Q_{{D_f}}}} \!\right) + {D_{{D_f}}}} \!\right)}^\infty  \!{{f_{{\rho _{{D_f}}}}}\!\int_0^\infty  \!{{f_{{\rho _{{g_{{D_f}}}}}{\rho _b}}}\left( y \right)} dxdy}\\
&\;=\! \frac{2}{{{\lambda _7}\!{\lambda _5}}}\int_0^\infty  \!{{e^{ \!- \frac{{{A_1}\!\left[ {y\left( {{A_{{D_f}}} + {Q_{{D_f}}}} \right) \!+\! {D_{{D_f}}}} \right]}}{{{\lambda _2}}}}}\!{K_0}\left( {2\sqrt {\frac{y}{{{\lambda _7}{\lambda _5}}}} } \right)} dy ,
\end{align}
where ${f_{{\rho _{{D_f}}}}} = \frac{1}{{{\lambda _2}}}{e^{ - \frac{x}{{{\lambda _2}}}}}$, ${f_{{\rho _{g{D_f}}}{\rho _b}}}\left( y \right) = \frac{{2{K_0}\left( {2\sqrt {\frac{y}{{{\lambda _7}{\lambda _5}}}} } \right)}}{{{\lambda _7}{\lambda _5}}}$, ${K_0}\left( x \right)$ is the modified Bessel function of the second kind. By using the integral equation \cite[Eq. (6.614)]{12}, we can obtain (11) after some mathematical manipulations.

\end{proof}

\begin{corollary}
At high SNRs, the asymptotic expression for the OP of $D_f$ of the AmBC NOMA system is given as
\begin{equation}
\label{12}
P_{out,\infty }^{{\Delta_1 }} = 1 + {\Delta_1 }{e^{{\Delta_1 } }}{\rm{Ei}}\left( { - {\Delta_1 }} \right),
\end{equation}
\end{corollary}

\emph{2) Outage Probability for $D_n$}

To successfully decode ${x_1}$ at $D_n$, two conditions are needed to be met: 1) $D_n$ can successfully decode ${x_2}$; 2) $D_n$ can successfully decode its own information ${x_1}$. Therefore, the OP of $D_n$ can be expressed as
\begin{equation}
\label{11}
P_{out}^{{D_n}} = 1 - {{\rm{P}}_r}\left( {\gamma _{{D_n}}^{{x_2}} > {\gamma _{th2}},\gamma _{{D_n}}^{{x_1}} > {\gamma _{th1}}} \right),
\end{equation}
where ${\gamma _{th1}}$ is the target rate at $D_n$.
\begin{theorem}
For Rayleigh fading channels, the analytical expression for OP of the near user can be obtained as
\begin{equation}
\label{15}
P_{out}^{{D_n}} = 1 + {\Delta_2 }{e^{{\Delta_2 } - \frac{{\varsigma {D_{{D_n}}}}}{{{\lambda _1}\gamma }}}}{\rm{Ei}}\left( { - {\Delta_2 }} \right),
\end{equation}
where $\varsigma  \!=\! \max \left\{ {\frac{{{\gamma _{th1}}}}{{{\xi _{{D_n}}}{a_1} - \left( {{a_2}{B_{{D_n}}} + {M_{{D_n}}}} \right){\gamma _{th1}}}},} \right.\!\!\;\;\left. {\frac{{{\gamma _{th2}}}}{{{\xi _{{D_n}}}{a_2} - {C_{{D_n}}}{\gamma _{th2}}}}} \right\}$, ${\Delta_2 } = \frac{{{\lambda _1}}}{{{\lambda _7}{\lambda _4}\varsigma \left( {{A_{{D_n}}} + {Q_{{D_n}}}} \right)}}$.
\end{theorem}

\begin{proof}
Following the similar derivation process of $P_{out}^{{D_f}}$, by substituting (\ref{5}), (\ref{6}) into (14), we can obtained $P_{out}^{{D_n}}$.
\end{proof}

\begin{corollary}
At high SNRs, the asymptotic expression for the OP of $D_f$ of the AmBC NOMA system is given as
\begin{equation}
\label{16}
P_{out,\infty }^{{D_n}} = 1 + {\Delta_2 }{e^{{\Delta_2 } }}{\rm{Ei}}\left( { - {\Delta_2 }} \right),
\end{equation}
\end{corollary}

\emph{3) Outage Probability for $BD$}

The BD signals can be successfully decoded when ${x_2}$ and ${x_1}$ are perfectly decode at $D_n$. Thus, the OP of BD can be expressed as
\begin{equation}
\label{11}
P_{out}^{BD} \!=\! 1 \!-\! {{\rm{P}}_r}\left( {\gamma _{{D_n}}^{{x_2}} > {\gamma _{th2}},\gamma _{{D_n}}^{{x_1}} > {\gamma _{th1}},\gamma _{{D_n}}^{c\left( t \right)} > {\gamma _{thc}}} \right),
\end{equation}
where ${\gamma _{thc}}$ is the target rate for decoding BD signals.
\begin{theorem}
For Rayleigh fading channels, we have

$\bullet$ Ideal conditions (${\varsigma _t} = {\varsigma _r} = 1$, ${\phi _t} = {\phi _r} = {0^ \circ }$)

For ideal conditions, the analytical expression for the OP of the BD in (15) is at the top of next page.
\begin{figure*}[!t]\label{18}
\normalsize
\begin{align}
\label{18}
P_{out}^{BD,id} = 1 + {\Delta_3 }{e^{{\Delta_3 } - \frac{{{\varsigma _{id}}}}{{\gamma {\lambda _1}}}}}{\rm{Ei}}\left( { - {\Delta_3 }} \right) + \frac{{\pi {\gamma _{thc}}}}{{N{\beta ^2}\gamma {\lambda _7}{\lambda _4}}}{e^{ - \frac{{{\varsigma _{id}}}}{{\gamma {\lambda _1}}}}}\sum\limits_{k = 0}^N {{e^{ - \frac{{{\gamma _{thc}}\left( {{\vartheta _k} + 1} \right)}}{{2{\lambda _1}}}}}{K_0}\left( {2\sqrt {\frac{{{\gamma _{thc}}\left( {{\vartheta _k} + 1} \right)}}{{2{\beta ^2}\gamma {\lambda _7}{\lambda _4}}}} } \right)\sqrt {1 - \vartheta _k^2} } ,
\end{align}
\hrulefill \vspace*{0pt}
\end{figure*}

\noindent{where ${\varsigma _{id}} \!=\! \max \!\left\{ {\frac{{{\gamma _{th2}}}}{{{a_2}  - {a_1} {\gamma _{th2}}}},\frac{{{\gamma _{th1}}}}{{{a_1} }}} \right\}$, ${\Delta_3 } \!= \!\frac{{{\lambda _1}}}{{{\lambda _7}{\lambda _4}{\varsigma _{id}}{\beta ^2}}}$, ${\vartheta _k} = \cos \!\!\left[ \!{\frac{\left( {2k - 1} \right)\pi}  {2N} } \!\right]$, $N$ is an accuracy-complexity tradeoff parameter.}

$\bullet$ Non-ideal conditions

For Non-ideal conditions, the analytical expression for the OP of BD in (16) is at the top of next page.
\begin{figure*}[!t]\label{19}
\normalsize
\begin{align}\nonumber
\label{19}
P_{out}^{BD,ni} =& 1 - {\Delta_4 }{e^{\frac{{{D_{{D_n}}}}}{{{\lambda _1}\left( {{B_{{D_n}}} + {M_{{D_n}}}} \right)\gamma }} + {\Delta_4 }}}{\rm{Ei}}\left( { - {\Delta_4 }} \right) + {\Delta_5 }{e^{{\Delta_5 } - \frac{{\varsigma {D_{{D_n}}}}}{{{\lambda _1}}}}}{\rm{Ei}}\left( { - {\Delta_5 }} \right) + {e^{ - \frac{{\varsigma {D_{{D_n}}}}}{{\gamma {\lambda _1}}}}}\frac{{\pi {\Delta_6 }}}{{N{\lambda _3}{\lambda _5}}}\sum\limits_{k = 0}^N {{e^{ - {\Delta_8 }}}{K_0}\left( {{\Delta_7 }} \right)\sqrt {1 - \vartheta _k^2} }  \\
&- {e^{\frac{{{D_{{D_n}}}}}{{{\lambda _1}\left( {{B_{{D_n}}} + {M_{{D_n}}}} \right)\gamma }}}}\frac{{\pi {\Delta_6 }}}{{N{\lambda _7}{\lambda _4}}}\sum\limits_{k = 0}^N {{e^{ - \frac{{{D_{{D_n}}}\left( {{\vartheta _k} + 1} \right)}}{{2{\lambda _1}\left( {{B_{{D_n}}} + {M_{{D_n}}}} \right)\gamma }}}}{K_0}\left( {{\Delta_7 }} \right)\sqrt {1 - \vartheta _k^2} }  ,
\end{align}
\hrulefill \vspace*{0pt}
\end{figure*}

\noindent{where ${\Delta_4 } = \frac{{{\lambda _1}\left( {{B_{{D_n}}} + {M_{{D_n}}}} \right){\gamma _{thc}}}}{{{\lambda _7}{\lambda _4}\left( {{Q_{{D_n}}} - {A_{{D_n}}}{\gamma _{thc}}} \right)}}$, ${\Delta_5 } = \frac{{{\lambda _1}}}{{{\lambda _7}{\lambda _4}\varsigma \left( {{A_{{D_n}}} + {Q_{{D_n}}}} \right)}}$, ${\Delta_6 } = \frac{{{D_{{D_n}}}{\gamma _{thc}}}}{{\left( {{Q_{{D_n}}}\gamma  - {A_{{D_n}}}\gamma {\gamma _{thc}}} \right)}}$, ${\Delta_7 } = 2\sqrt {\frac{{\left( {{\vartheta _k} + 1} \right){R_6}}}{{2{\lambda _7}{\lambda _4}}}} $, ${\Delta_8 } = \frac{{\varsigma \left( {{A_{{D_n}}} + {Q_{{D_n}}}} \right)\left( {{\vartheta _k} + 1} \right){R_6}}}{{2{\lambda _1}\gamma }}$.}
\end{theorem}

\begin{proof} For ideal conditions, substituting ${\varsigma _t} = {\varsigma _r} = 1$ and ${\phi _t} = {\phi _r} = {0^ \circ }$ into (\ref{5}), (\ref{6}), (\ref{7}), then according to (14), $P_{out}^{BD,id}$ can be expressed as
\begin{align}\nonumber
\label{19}
 &P_{out}^{BD,id}= \int_{\frac{{{\gamma _{thc}}}}{{{\beta ^2}\gamma }}}^\infty  {{e^{ - \frac{{{\varsigma _{id}}\left( {y{\beta ^2} + 1} \right)}}{{{\lambda _1}}}}}\frac{2}{{{\lambda _7}{\lambda _4}}}{K_0}\left( {2\sqrt {\frac{y}{{{\lambda _7}{\lambda _4}}}} } \right)} dy\\\nonumber
&\;\;\;\;\;\;\;\;\;\;\;= \int_0^\infty  {{e^{ - \frac{{{\varsigma _{id}}\left( {y{\beta ^2} + 1} \right)}}{{{\lambda _1}}}}}\frac{2}{{{\lambda _7}{\lambda _4}}}{K_0}\left( {2\sqrt {\frac{y}{{{\lambda _7}{\lambda _4}}}} } \right)} dy \\
 &\;\;\;\;\;\;- \underbrace {\int_0^{\frac{{{\gamma _{thc}}}}{{{\beta ^2}\gamma }}} {{e^{ - \frac{{{\varsigma _{id}}\left( {y{\beta ^2} + 1} \right)}}{{{\lambda _1}}}}}\frac{2}{{{\lambda _7}{\lambda _4}}}{K_0}\left( {2\sqrt {\frac{y}{{{\lambda _7}{\lambda _4}}}} } \right)} }_ady,
\end{align}
where $a$ can be approximated by Gaussian-Chebyshev quadrature \cite{13}. Utilizing \cite[Eq. (6.611)]{12} to the first term of (17), we can obtain the result of (18) after some mathematical manipulations. Similarly, $P_{out}^{BD,ni}$ can be obtained.

\end{proof}

\begin{corollary}
At high SNRs, the asymptotic expressions for the OP of $BD$ for the AmBC NOMA system can be given by

$\bullet$ Ideal conditions
\begin{equation}
\label{11}
P_{out,\infty }^{BD,id} = 1 + {\Delta_3 }{e^{{\Delta_3 }}}{\rm{Ei}}\left( { - {\Delta_3 }} \right),
\end{equation}

$\bullet$ Non-ideal conditions
\begin{equation}
\label{11}
P_{out,\infty }^{BD,ni} = 1 - {\Delta_4 }{e^{{\Delta_4 }}}{\rm{Ei}}\left( { - {\Delta_4 }} \right) + {\Delta_5 }{e^{{\Delta_5 }}}{\rm{Ei}}\left( { - {\Delta_5 }} \right).
\end{equation}

\end{corollary}

Furthermore, the diversity order is investigated, which can be defined as:
\begin{equation}\label{18}
d =  - \mathop {\lim }\limits_{\gamma  \to \infty } \frac{{\log \left( {P_{out}^\infty } \right)}}{{\log {\gamma}}},
\end{equation}

\begin{corollary}
The diversity order of $D_f$, $D_n$ and $BD$ are given as:
\begin{equation}\label{21}
{d_{{D_f}}} = {d_{{D_n}}} = d_{BD}^{id} = d_{BD}^{ni}.
\end{equation}

\end{corollary}

\begin{remark}
From \textbf{Corollary 1}, \textbf{Corollary 2}, \textbf{Corollary 3} and \textbf{Corollary 4}, we can see that when transmit SNR goes to infinity, the asymptotic outage performance of the $D_f$, $D_n$ and $BD$ become a constant. It means that the OP exists error floor, which results in $0$ diversity order. This error floor is determined by the parameters of IQI and the reflection coefficient $\beta $.
\end{remark}

\subsection{Intercept Probability Analysis}

User $j$, $j \in \{ {{D_f},{D_n},BD} \}$ will be intercepted if $E$ can successfully wiretap $j$'s signal, i.e. ${\gamma _E^{p} > {\gamma _{thE,j}}}$, $p \in {\rm{ }}\left\{ {{x_2},{x_1},c\left( t \right)} \right\}$. Thus, the IP of $j$ by $E$ can be expressed as
\begin{equation}
\label{11}
P_{int}^{j} = {{\rm{P}}_r}\left( {\gamma _E^{p} > {\gamma _{thE,j}}} \right),
\end{equation}
where ${\gamma _{thE,j}}$ is the secrecy SNR threshold for $j$.
\begin{theorem}
For Rayleigh fading channels, the analytical expressions for the IP of the far user, the near user and BD can be respectively obtained as

For the far user and near user, we have
\begin{equation}
\label{11}
P_{int}^{{D_f}} =  - {\Delta_9 }{e^{{\Delta_9 } - \frac{{{D_E}{A_2}}}{{{\lambda _2}}}}}{\rm{Ei}}\left( { - {\Delta_9 }} \right),
\end{equation}
\begin{equation}
\label{11}
P_{int}^{{D_n}} \!= \! - {\Delta_{10} }{e^{{\Delta_{10} } \!- \frac{{{D_E}{\gamma _{thE1}}}}{{{\lambda _3}\left( {{{\left| {{\mu _{{r_E}}}} \right|}^2}{{\left| {{\mu _{{t_s}}}} \right|}^2}{a_1}\gamma  - \!\left( {{a_2}{B_E} \!+\! {M_E}} \right)\gamma {\gamma _{thE1}}} \right)}}}}{\rm{Ei}}\left( { \!-\! {\Delta_{10} }} \right),
\end{equation}
where ${\Delta_9 } = \frac{{{\lambda _4}}}{{{\lambda _7}{\lambda _6}{A_2}\left( {{A_E} + {Q_E}} \right)}}$, ${A_2} = \frac{{{\gamma _{thE2}}}}{{{\xi _E}{a_2} - {C_E}{\gamma _{thE2}}}}$, ${\Delta _{10}} = \frac{{{\lambda _3}\left( {{{\left| {{\mu _{{r_E}}}} \right|}^2}{{\left| {{\mu _{{t_s}}}} \right|}^2}{a_1} - \left( {{a_2}{B_E} + {M_E}} \right){\gamma _{thE1}}} \right)}}{{{\lambda _7}{\lambda _6}\left( {{A_E} + {Q_E}} \right){\gamma _{thE1}}}}$.

For the BD, we have

$\bullet$ Ideal conditions (${\varsigma _t} = {\varsigma _r} = 1$, ${\phi _t} = {\phi _r} = {0^ \circ }$)
\begin{equation}
\label{11}
P_{int}^{BD,id} \!=\! 1 \!-\! \frac{{\pi {\gamma _{thE3}}}}{{N{\lambda _7}{\lambda _6}{\beta ^2}\gamma }}\!\sum\limits_{k = 0}^N \!{{K_0}\!\left( \!{2\sqrt {\frac{{{\gamma _{thE3}}\left( \!{{\vartheta _k} + 1} \!\right)}}{{2{\lambda _7}{\lambda _6}{\beta ^2}\gamma }}} } \!\right)\!\sqrt {1\! -\! \vartheta _k^2} } ,
\end{equation}

$\bullet$ Non-ideal conditions
\begin{figure*}[!t]\label{18}
\normalsize
\begin{align}\nonumber
\label{18}
P_{out}^{BD,ni} =& 1 - \frac{{\pi {\Delta _{11}}}}{N}\sum\limits_{k = 0}^N {{K_0}\left( {{\Delta _{13}}} \right)\sqrt {1 - \vartheta _k^2} }  + {\Delta _{12}}{e^{\frac{{{D_E}}}{{{\lambda _3}\left( {{B_E} + {M_E}} \right)\gamma }} + {\Delta _{12}}}}{\rm{Ei}}\left( { - {\Delta _{12}}} \right) \\
&+ {e^{\frac{{{D_E}}}{{{\lambda _3}\left( {{B_E} + {M_E}} \right)\gamma }}}}\frac{{\pi {\Delta _{11}}}}{N}\sum\limits_{k = 0}^N {{e^{ - \frac{{{D_E}\left( {{\vartheta _k} + 1} \right)}}{{2{\lambda _3}\left( {{B_E} + {M_E}} \right)\gamma }}}}{K_0}\left( {{\Delta _{13}}} \right)\sqrt {1 - \vartheta _k^2} }  ,
\end{align}
\hrulefill \vspace*{0pt}
\end{figure*}

For Non-ideal conditions, the analytical expression for the IP of BD in (26) is at the top of next page.

\noindent{where ${\Delta _{11}} = \frac{{{D_E}{\gamma _{thE3}}}}{{{\lambda _7}{\lambda _6}\left( {{Q_E}\gamma  - {A_E}\gamma {\gamma _{thE3}}} \right)}}$, ${\Delta _{12}} = \frac{{{\lambda _3}\left( {{B_E} + {M_E}} \right){\gamma _{thE3}}}}{{{\lambda _7}{\lambda _6}\left( {{Q_E} - {A_E}{\gamma _{thE3}}} \right)}}$, ${\Delta _{13}} = 2\sqrt {\frac{{{\Delta _{11}}\left( {{\vartheta _k} + 1} \right)}}{2}} $.}
\end{theorem}

\begin{remark}
From \textbf{Theorem 4}, We can obtain that as the reflection coefficient $\beta $ increases, $P_{int}^{{D_f}}$ and $P_{int}^{{D_n}}$ decreases, while $P_{int}^{BD}$ increases; Moreover, IQI parameters has a beneficial effect on the IPs for the far user, the near user and the BD.
\end{remark}

\begin{figure}[!t]
\setlength{\abovecaptionskip}{0pt}
\centering
\includegraphics [width=3.5in]{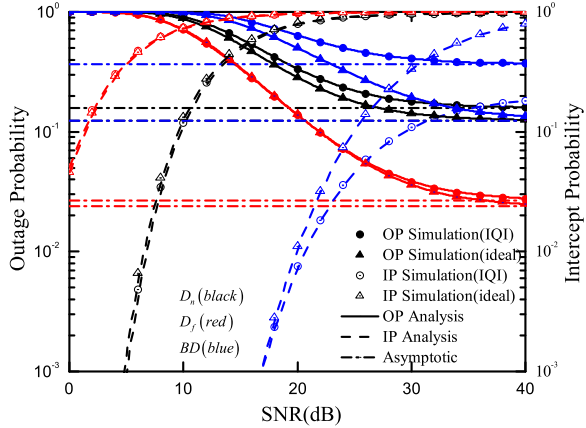}
\caption{{ OP and IP versus the transmit SNR}}
\label{fig3}
\end{figure}
\section{Numerical Results}
This section provides some numerical results to validate the correctness of the analysis in Section III. The results are verified by using Monte Carlo simulations with ${10^6}$ trials. Unless otherwise stated, we have the following settings: ${N_0} = 1$, ${\gamma _{thf}} = 1$, ${\gamma _{thn}} = 2$, ${\gamma _{thc}} = 0.1$, ${\gamma _{thE,{D_n}}} = 1$, ${\gamma _{thE,{D_f}}} = 1.2$, ${\gamma _{thE,{BD}}} = 0.8$; The power allocation coefficients are ${a_1} =0.1$ and ${a_2} = 0.9$; The channel fading parameters are ${\lambda _1} ={\lambda _3}= 1$, ${\lambda _2} = 0.1$, ${\lambda _4} = 0.5$, ${\lambda _5} = 0.8$, ${\lambda _6} = 0.2$, ${\lambda _7} = 0.1$. The reflection coefficient is $\beta  = 0.1$, while for the ideal RF front-end, ${\varsigma _t} = {\varsigma _r} = 1$, ${\phi _t} = {\phi _r} = 0$.

Fig. 2 plots the OP and IP versus transmit SNR under ideal and IQI conditions with ${\varsigma _t} = {\varsigma _r} = 1.05$, ${\phi _t} = {\phi _r} = {20^ \circ }$. The theoretical results are in excellent agreement with the equivalent simulated results. A specific observation can be obtained that IQI have significant effect on the near user, this happens because that the near user suffers from SIC, which is impaired by IQI. The changes of OP and IP for BD are very obvious due to the effects of two time SIC impired by IQI, which indicate that BD performance is more sensitive to IQI at high SNRs. On the other hand, although IQI reduces the reliability of the considered system, it enhances the security to some extent. Finally, we can also see that there exists a trade-off between reliability and security, that is, when the outage performance is relaxed, IP can be enhanced, and vice versa.

\begin{figure}[!t]
\setlength{\abovecaptionskip}{0pt}
\centering
\includegraphics [width=3.5in]{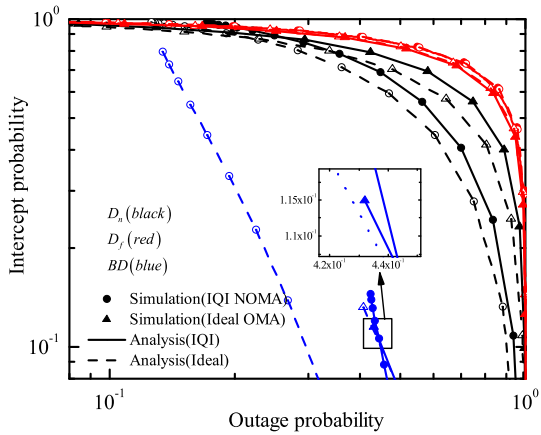}
\caption{{IP versus OP for the $D_f$, $D_n$ and BD }}
\label{fig2}
\end{figure}
\begin{figure}[!t]
\setlength{\abovecaptionskip}{0pt}
\centering
\includegraphics [width=3.5in]{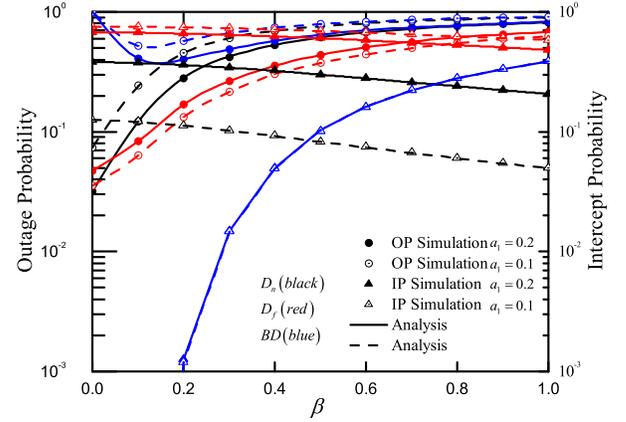}
\caption{{OP and IP versus the reflection coefficient $\beta $ for different power allocation parameter $a_1$}}
\label{fig4}
\end{figure}

Fig. 3 demonstrates the impact of IP versus OP under ideal/non-ideal conditions ${\varsigma _t} = {\varsigma _r} = \{1, 1.05\}$, ${\phi _t} = {\phi _r} = \{0^ \circ , 25^ \circ \}$. For the purpose of comparison, the curves of OMA are provided with $\gamma _{thc}^{OMA} = \gamma _{thE,BD}^{OMA} = 1.8$.\footnote{The reason for this parameter setting is that when the parameter of NOMA is set to $\gamma _{thc}^{OMA} = \gamma _{thE,BD}^{OMA} = 1.8$, which lead to the OP of BD is 1, and the IP is too small.}  These results show that for a given OP, as IQI increases the corresponding IP reduces, which leads the decrease of the secrecy performance. In addition, we can observe that the effect of IQI on each user is different. This effect depends on the decoding order of the users and the system parameter settings, in which the far user is the least affected and experiences almost no performance degradation caused by the presence of IQI. BD changes is most significant.  We can also observe that for a given OP, the secure performance of NOMA is better than that of OMA system for the near user, the performance of the far user is the opposite, this is due to the large power allocation parameter of the far user. Finally, it should be noted that the IP of the BD is the smallest, i.e., the BD has the best secure performance.

Fig.4 presents the OP and IP versus the reflection coefficient $\beta $ for different power allocation parameter $a_1$ with ${\varsigma _t} = {\varsigma _r} = 1.1$, ${\phi _t} = {\phi _r} = {5^ \circ }$, where the transmit SNR for OP and IP are $SNR=25dB$, $SNR=10dB$. From these curves, we can observe that $P_{out}^{{D_f}}$ and $P_{out}^{{D_n}}$ increases, $P_{int}^{{D_f}}$ and $P_{int}^{{D_n}}$ decrease with the increase of $\beta $. This happens because that when $\beta $ becomes large, the interference of the backscatter link increases, which reduces the reliability of $D_f$ and $D_n$, and improve the security. The OP of BD decreases first and then increases as $\beta $ increases. This happens because that when $\beta $ is small, it is easier for $D_n$ to decode its own information successfully, but it is difficult to decode BD signals. When $\beta $ is large, it is difficult for $D_n$ to decode its own information successfully. Noted that when $SNR=25dB$, the outage performance is optimal when $\beta=0.14$ for $a_1=0.2$, $\beta=0.12$ for $a_1=0.1$. This means that the optimal reflection coefficient is determined by power allocation coefficient.

\section{Conclusion}
This paper investigates the impact of IQI on the reliability and security of the AmBC NOMA system. This is carried out by deriving closed-form analytical expressions of the OP and IP in the presence of IQI, and the asymptotic behavior in the high SNR regime and diversity order for the OP are analyzed. The simulation results show that although IQI reduces the reliability for the far user, the near user and the BD, while it also improves the security of the the three devices. In addition, the increase of $\beta$ will increase the OP and decrease the IP for far user and near user. Finally, we can conclude that BD of the AmBC NOMA system has better secrecy performance, which drives the application in the IoT networks.

\bibliographystyle{IEEEtran}
\bibliography{ZML_tvt}

\end{document}